\documentclass[sigconf,natbib=false]{acmart}
\AtBeginDocument{%
  }

\setcopyright{acmlicensed}
\copyrightyear{2026}
\acmYear{2026}
\acmDOI{XXXXXXX.XXXXXXX}
\RequirePackage[
  datamodel=acmdatamodel,
  style=acmnumeric,
  ]{biblatex}

\usepackage{algorithm}
\usepackage{graphicx}
\usepackage{algorithmic}
\usepackage{enumitem}

\usepackage{amsmath}
\usepackage{multirow}

\begin{document}

\title{JD-BP: A Joint-Decision Generative Framework for Auto-Bidding and Pricing}


\settopmatter{authorsperrow=4}
\author{Linghui Meng}
\authornote{Both authors contributed equally to this research.}
\email{menglinghui1@jd.com}
\affiliation{%
  \institution{JD.com}
  \country{China}
}

\author{Chun Gan}
\authornotemark[1]
\authornote{Corresponding author.}
\email{ganchun1@jd.com}
\affiliation{%
  \institution{JD.com}
  \country{China}
}

\author{Shengsheng Niu}
\email{niushengsheng@jd.com}
\affiliation{%
  \institution{JD.com}
  \country{China}
}

\author{Chengcheng Zhang}
\authornote{This work was performed during an internship at JD.com.}
\email{224010028@link.cuhk.edu.cn}
\affiliation{%
  \institution{The Chinese University of Hong Kong}
  \country{China}
}

\author{Chenchen Li}
\email{lichenchen31@jd.com}
\affiliation{%
  \institution{JD.com}
  \country{China}
}

\author{Chuan Yang}
\email{yangchuan56@jd.com}
\affiliation{%
  \institution{JD.com}
  \country{China}
}

\author{Yi Mao}
\email{maoyi3@jd.com}
\affiliation{%
  \institution{JD.com}
  \country{China}
}

\author{Xin Zhu}
\email{zhuxin3@jd.com}
\affiliation{%
  \institution{JD.com}
  \country{China}
}

\author{Jie He}
\email{hejie67@jd.com}
\affiliation{%
  \institution{JD.com}
  \country{China}
}

\author{Zhangang Lin}
\email{linzhangang@jd.com}
\affiliation{%
  \institution{JD.com}
  \country{China}
}

\author{Ching Law}
\email{lawching@jd.com}
\affiliation{%
  \institution{JD.com}
  \country{China}
}

\renewcommand{\shortauthors}{Linghui Meng, Chun Gan, Shengsheng Niu, Chengcheng Zhang et al.}

\begin{abstract}

Auto-bidding services optimize real-time bidding strategies for advertisers under key performance indicator (KPI) constraints such as target return on investment and budget. 
However, uncertainties such as model prediction errors and feedback latency can cause bidding strategies to deviate from ex-post optimality, leading to inefficient allocation. 
To address this issue, we propose \textbf{JD-BP}, a \textbf{J}oint generative \textbf{D}ecision framework for \textbf{B}idding and \textbf{P}ricing. 
Unlike prior methods, JD-BP jointly outputs a bid value and a pricing correction term that acts additively with the payment rule such as GSP. 
To mitigate adverse effects of historical constraint violations, we design a memory-less Return-to-Go that encourages bidding actions to focus on future value maximization, while the accumulated bias is handled by the pricing correction. 
Moreover, a trajectory augmentation algorithm is proposed to generate joint bidding-pricing trajectories from a
(possibly arbitrary) base bidding policy, enabling efficient plug-and-play deployment of our algorithm from existing RL/generative bidding models.
Finally, we employ an Energy-Based Direct Preference Optimization method in conjunction with a cross-attention module to enhance the joint learning performance of bidding and pricing correction.
Offline experiments on the AuctionNet dataset demonstrate that JD-BP achieves state-of-the-art performance. 
Online A/B tests on a major e-commerce platform demonstrate significant improvements: a 4.70\% increase in ad revenue and a 6.48\% increase in target cost.
\end{abstract}




\maketitle

\section{Introduction}

Auto-bidding services have become integral to modern online advertising ecosystems---one of the largest-scale decision-making and data-mining applications on the Web---acting as intelligent agents that place bids on behalf of advertisers to maximize value under KPI constraints such as target cost-per-click (CPC) and return on investment (ROI)~\cite{Evans2009TheOA,Wang2015RealTimeBA,Lv2022UtilityMO,Balseiro2021TheLO}. From a mathematical perspective, auto-bidding problem is typically formulated as a constrained optimization task where the objective is to maximize total advertiser value subject to predefined KPI constraints. When future values and the advertising environment are known, exact optimal bidding formula can be derived through linear programming formulations and Lagrangian duality methods~\cite{aggarwal2019autobidding,he2021unified}.

However, the practical implementation of auto-bidding faces significant challenges due to the stochastic and dynamic nature of advertising environments. Two primary sources of uncertainty dominate: (1) model prediction errors in estimating click-through rates (CTR) and conversion rates (CVR), and (2) feedback latency where conversion events may be observed hours or days after the initial impression. These uncertainties necessitate continuous adjustment of bidding strategies based on real-time performance feedback, transforming the problem into an online decision-making process with partial and delayed observability.

To address this adaptive control challenge, the research community has explored various approaches. Early work employed classical control methods such as Proportional-Integral-Derivative (PID) controllers and model predictive control~\cite{zhang2022control}. More recently, reinforcement learning (RL) techniques have gained prominence for their ability to learn optimal policies through interaction with the environment~\cite{mou2022sustainable,Cai2017RealTimeBB,Ye2020DeepRL}. The latest advances leverage generative models, including Decision Transformers (DT)~\cite{Chen2021DecisionTR,Gao2025GenerativeAW,Li2024GASGA} and diffusion models~\cite{Guo2024AIGBGA,Li2025GenerativeAI,Peng2025ExpertGuidedDP}, which can capture complex temporal dependencies in historical bidding data and generate more robust policies.

Despite these methodological advancements, a fundamental limitation persists in current approaches:  the bidding action is learned to simultaneously achieve two distinct objectives—value maximization and constraint satisfaction. This coupling creates inherent inefficiencies when KPI constraints are violated during intermediate steps. Consider a scenario where an advertiser sets the target ROI to 10 while the observed real-time ROI at a certain timestep equals 6 due to either conversion delays or prediction inaccuracies. Conventional auto-bidding algorithms typically respond by reducing future bid values to mitigate constraint violation risks~\cite{Zhang2014OptimalRB,Wu2018BudgetCB}. This conservative adjustment, while protecting against further constraint violations, may cause the advertiser to miss valuable conversion opportunities. More critically, such bidding strategy could lead to allocation inefficiency at the market level, where agents with lower true values may outbid those with genuinely higher values, distorting the auction mechanism’s intended efficiency properties.

The core problem lies in the temporal misalignment between bidding actions and their consequences. When an agent attempts to compensate for historical constraint violations through future bidding decisions, it creates a feedback loop that distorts the relationship between bid values and true advertiser valuations. Such misalignment is particularly problematic in second-price auction environments where the pricing mechanism depends on competitors’ bids, creating complex strategic interactions that single-action approaches cannot adequately address.

To overcome these limitations, we propose a paradigm shift from single-action to dual-action optimization. Our key insight is that value maximization and constraint compensation should be decoupled into separate but coordinated actions. We introduce a novel joint optimization framework that incorporates both bidding actions (controlling auction participation) and pricing actions (adjusting effective costs through mechanism correction). The pricing correction term acts additively on the underlying auction mechanism (such as Generalized Second Price), allowing the system to address historical constraint violations without distorting current bidding strategies that should reflect true advertiser valuations.

We implement this framework through \textbf{JD-BP} (Joint Decision framework for Bidding and Pricing), a generative decision-making model based on the Decision Transformer architecture. The model features several innovative design elements: First, we introduce a pricing correction term that operates additively during the payment settlement phase. To effectively learn the joint bidding and pricing decisions, we design a memory-less Return-to-Go (RTG) that excludes past constraint violation signals from the bidding decision process, ensuring bidding actions focus on future value maximization. We further incorporate a cross-attention module to enable pricing adjustments to perceive bidding actions. Since initial deployment lacks training data, we develop a trajectory augmentation algorithm that generates joint bidding-pricing trajectories from a (possibly arbitrary) base bidding policy, enabling efficient plug-and-play deployment of our algorithm from existing bidding models. Finally, we propose an energy-based Direct Preference Optimization (DPO) fine-tuning~\cite{Rafailov2023DirectPO} using trajectories that pair high-reward outcomes with appropriate bidding-pricing action combinations, allowing our model to learn preference distinctions between different types of corrective actions.

Our contributions are summarized as follows:
\begin{itemize}[leftmargin=*, topsep=0pt, partopsep=0pt, parsep=0pt, itemsep=2pt]
    \item We propose \textbf{JD-BP}, a novel joint optimization framework decoupling value maximization (bidding) from historical constraint compensation (pricing). To the best of our knowledge, this is the first generative framework treating pricing mechanisms as an actionable dimension in auto-bidding.
    \item Rather than simply applying standard generative models, we introduce a \textbf{memory-less RTG} to decouple the sequential credit assignment, and design a \textbf{Gate-Selected Cross-Attention (GCA)} architecture to ensure pricing actions are causally conditioned on bidding outcomes.
    \item To overcome the limitation of categorical-based Direct Preference Optimization (DPO), we formulate an \textbf{Energy-Based Continuous DPO} customized for deterministic regression tasks in auto-bidding, allowing the model to distinguish and align with high-reward trajectories without artificial distribution assumptions.
    \item In addition to offline experimental validation achieving state-of-the-art results, we deploy JD-BP on a leading global e-commerce platform, resulting in a 4.70\% increase in ad revenue and a 6.48\% increase in target cost.
\end{itemize}

\section{Preliminary}
In this section, we first present the classical formulation of an auto-bidding task as a constrained optimization problem. 
We then discuss the closely related field of online decision-making and its relevance to the solution of the auto-bidding problem.

\subsection{Auto-Bidding Problem}
Consider the scenario where \textit{N} bidding opportunities arrive sequentially throughout the day, the classic auto-bidding problem can be formulated as follows:
\begin{equation} \label{ori-pro}
    \begin{aligned}
        \max_{x_i}\ & \textstyle\sum_i^N x_i v_i \\
        \text{s.t.}\ & \textstyle\sum_i^N x_i c_i \leq B, \quad x_i \in \{0,1\},\ \forall i, \\
         & \textstyle\sum_i^N x_i c_i \leq \rho_j \textstyle\sum_i^N x_i \gamma_{ij},\ \forall j.
    \end{aligned}
\end{equation}
where $v_i$ represents the advertiser's estimated value for the i-th impression, $c_i$ denotes actual payment (i.e., cost) charged to the advertiser and $x_i$ indicates whether the advertiser has won this impression opportunity under a given bidding mechanism such as the second-price auction, the most common mechanism in industrial practice.
The model considers two types of constraints to ensure performance. One is the budget constraint, where $B$ represents the total advertising budget. The second type pertains to KPI constraints, such as CPC and ROI constraints. Here $\gamma_{ij}$ denotes the value of the $j$-th KPI metric associated with impression $i$ (e.g., the click indicator for CPC or the conversion value for ROI), and $\rho_j$ is the corresponding target bound; for the ROI constraint, $\gamma_{ij}=v_i$ and $\rho_j$ is the target cost-to-value ratio.

To solve this large-scale 0-1 programming problem, prior work typically relaxes the integer constraint on $x_i$, transforming it into a linear program. Moreover, in practical industrial systems, budget pacing is often handled by a separate control module, allowing the optimization to focus primarily on KPI constraints. When future values and the advertising environment are known, the optimal bidding strategy can be derived as:
\begin{equation}
\label{opt-bid}
bid_i = \lambda_0 v_i +  \sum_j \rho_j \lambda_j \gamma_{ij}
\end{equation}
where $\lambda_0$ and $\lambda_j$ denote the dual variables associated with the budget constraint and the KPI constraints, respectively~\cite{aggarwal2019autobidding,he2021unified}.

Without loss of generality, the subsequent discussion will focus on the problem considering both budget and ROI constraints, which is the most prevalent scenario in industrial practice.

\subsection{Online Decision-Making for Auto-Bidding}

The closed-form optimal solution given by Eq.~\eqref{opt-bid} depends on two essential assumptions: 1) knowing its future values $\{v_i\}$ and 2) the opponents' values and bidding strategies are stationary~\cite{Jin2018RealTimeBW}.
However, the advertising environment of online advertising is highly competitive and uncertain, mainly due to the following three reasons.

First, traffic distribution may shift significantly due to factors such as promotions, weather events, or other external stimuli.
Second, conversion latency is a common phenomenon in e-Commerce that could bring challenge to online decision-making. This refers to the delay between a user’s click and subsequent conversion (e.g., purchase)~\cite{Chapelle2014ModelingDF,Chan2023CapturingCR}. 
Last but not least, despite improvements in CTR and CVR prediction models, the prediction error of machine learning models always exists~\cite{Richardson2007PredictingCE}. 
Consequently, auto-bidding is often modeled as an online Decision-Making problem in practice~\cite{Amin2012BudgetOF}.

To characterize the gap between the posterior outcome and the optimal outcome in hindsight, regret is measured along two dimensions: value maximization and constraint violation.

Formally, 
let $\mathcal{E}$
be the advertising environment
and $T$ be the total timesteps, respectively.
For each time step $t \in \{1, 2, \cdots, T\}$, the auto-bidding agent receives a state $s_t = [h_t, c_t, x_t]\in \mathcal{S}$, where $h_t$ contains information from its bidding history such as cost, remaining budget, $c_t$ represents target constraints and $x_t$ describes the traffic distribution.
Let $\{x_t^\star\}$ be the optimal allocation in hindsight. The two regret dimensions can then be made precise: a \emph{value regret} $\mathcal{R}_v=\sum_t x_t^\star v_t-\sum_t x_t v_t$, measuring the value forgone relative to the hindsight optimum, and a \emph{constraint-violation regret} $\mathcal{R}_c=\max\!\big(0,\ \sum_t x_t c_t-\rho\sum_t x_t v_t\big)$, penalizing cost that exceeds the target ratio $\rho$. An ideal policy should drive both to zero simultaneously. In practice, however, the two are tightly entangled: when a historical violation has already inflated $\mathcal{R}_c$, a single-action bidder can only suppress it by bidding conservatively, which in turn sacrifices future value and enlarges $\mathcal{R}_v$. This intrinsic coupling between value maximization and constraint compensation is precisely the tension that the joint bidding--pricing design developed next is meant to resolve.

\section{Methodology}
In this chapter, we first present the mathematical formulation for the joint modeling of bidding and pricing, along with the optimal closed-form solution under perfect information of future bidding opportunities. Given the powerful sequential decision-making capabilities of the Decision Transformer, we adopt it as our backbone architecture to solve this joint decision-making problem. We propose a memoryless offline trajectory generation method to produce high-quality trajectories with additional pricing action for model training. Notably, this method is applicable to all baseline models that consider only bidding actions. Finally, we introduce a DPO fine-tune module to further enhance the model's performance. The overall algorithm architecture is illustrated in Fig.~\ref{fig:algorithm_stru}.

\subsection{Joint Decision Framework}
\label{sec:joint}

Following the preliminaries in Eq.~\eqref{ori-pro}, we index the future bidding opportunities by $t=t_m,\dots,T$, where $t_m$ is the current decision step; $x_t\in\{0,1\}$, $v_t$, $c_t$ retain the meanings defined in Eq.~\eqref{ori-pro}, and $B_{t_m}$ is the remaining budget at $t_m$. The core of our framework is an additive \emph{pricing correction} $y_t\le 0$ applied at settlement, so that the actual charge for a won opportunity is $p_t=c_t+y_t$.

\noindent\textbf{Historical deficit.}
Due to prediction errors and delayed feedback, the ROI constraint may be violated over the elapsed steps $t<t_m$. We quantify this violation by the \emph{historical deficit}
\begin{equation}\label{eq:bal}
bal=\max\!\Big(0,\ \textstyle\sum_{t<t_m} x_t c_t-\rho\sum_{t<t_m} x_t v_t\Big)\ \ge 0,
\end{equation}
i.e.\ the amount by which the historical cost \emph{exceeds} the level $\rho\sum_{t<t_m} x_t v_t$ permitted by the target ratio $\rho$. At step $t_m$ the history is fully observed, so $bal$ is a known non-negative constant.

\noindent\textbf{Joint problem.}
Conditioned on $bal$, we jointly optimize the future bids $\{x_t\}$ and pricing corrections $\{y_t\}$:
\begin{equation}\label{jom}
    \begin{aligned}
        \max_{x_t,\,y_t}\ & \textstyle\sum_{t=t_m}^{T} x_t v_t \\
        \text{s.t.}\ & \textstyle\sum_{t=t_m}^{T} x_t (c_t+y_t)\le B_{t_m} \\
         & \textstyle\sum_{t=t_m}^{T} x_t c_t\ \le\ \rho\sum_{t=t_m}^{T} x_t v_t \\
         & \textstyle\sum_{t=t_m}^{T} x_t y_t + bal = 0 \\
         & x_t\in\{0,1\},\ y_t\le 0,\ \forall t\ge t_m .
    \end{aligned}
\end{equation}
The first constraint caps the \emph{actual} expenditure (after correction) by the remaining budget; the second enforces the \emph{future} ROI target; the third requires the corrections to refund exactly the historical deficit, i.e.\ $\sum_{t\ge t_m} x_t y_t=-bal\le 0$ (a net refund). Intuitively, the pricing channel absorbs the historical violation so that bidding is free to maximize future value.

\begin{theorem}\label{thm:equiv}
Given the future opportunities after $t_m$, problem~\eqref{jom} is equivalent to the single-action bidding problem
\begin{equation}\label{eq:reduced}
\begin{aligned}
\max_{x_t}\ & \textstyle\sum_{t=t_m}^{T} x_t v_t \\
\text{s.t.}\ & \textstyle\sum_{t=t_m}^{T} x_t c_t \le B_{t_m}+bal \\
& \textstyle\sum_{t=t_m}^{T} x_t c_t \le \rho \sum_{t=t_m}^{T} x_t v_t,\quad x_t\in\{0,1\}.
\end{aligned}
\end{equation}
Problem~\eqref{eq:reduced} shares the same form as the classic bidding problem~\eqref{ori-pro}; hence its optimal policy retains the closed-form bid of Eq.~\eqref{opt-bid}, $bid_t=\hat{\lambda}_0 v_t+\sum_j \rho_j \hat{\lambda}_j \gamma_{tj}$, and a feasible pricing rule is the uniform refund $y_t=-\,bal\big/\!\sum_{t=t_m}^{T} x_t$.
\end{theorem}

\begin{proof}
Substituting the deficit-compensation constraint $\sum_{t\ge t_m} x_t y_t=-bal$ into the budget constraint of~\eqref{jom} eliminates $y_t$:
\begin{equation}
\begin{aligned}
\sum_{t\ge t_m} x_t(c_t+y_t)\le B_{t_m}
&\ \Longleftrightarrow\ \sum_{t\ge t_m} x_t c_t - bal \le B_{t_m} \\
&\ \Longleftrightarrow\ \sum_{t\ge t_m} x_t c_t \le B_{t_m}+bal,
\end{aligned}
\nonumber
\end{equation}
which yields~\eqref{eq:reduced}. We stress that the augmented bound $B_{t_m}+bal$ does \emph{not} grant extra spending power: since the corrections refund exactly $bal$, the actual expenditure $\sum x_t(c_t+y_t)=\sum x_t c_t-bal$ still stays within the true budget $B_{t_m}$. The term $bal$ only reflects that, by deferring the historical over-payment to the pricing channel, the bidding policy is \emph{relieved} of compensating it and may bid as if it had $bal$ extra cost headroom.

Since~\eqref{eq:reduced} differs from~\eqref{ori-pro} only by the relaxed budget bound $B_{t_m}+bal$ and is linear in $x_t$, the Lagrangian-duality argument behind Eq.~\eqref{opt-bid} applies unchanged. Therefore the optimal bidding policy keeps the closed form of Eq.~\eqref{opt-bid}; only the dual variables differ (denoted $\hat{\lambda}_0,\hat{\lambda}_j$) because the budget is enlarged from $B_{t_m}$ to $B_{t_m}+bal$. Finally, any non-positive $\{y_t\}$ summing to $-bal$ over the won set satisfies the remaining constraints; the uniform refund $y_t=-bal/\sum_{t=t_m}^{T} x_t\le 0$ is one such choice.
\end{proof}

\begin{corollary}\label{cor:gap}
Let $V_O^\star$ and $V_J^\star$ be the optimal values of the original online bidding problem (without pricing correction) and the joint problem~\eqref{jom}. Then $V_O^\star\le V_J^\star$.
\end{corollary}

\begin{proof}
Without a pricing channel, the historical deficit must be repaid through future bidding, so the original problem carries the tightened ROI constraint
\begin{equation}\label{eq:orig}
\begin{aligned}
&\sum_{t\ge t_m} x_t c_t\le B_{t_m}, \\
&\sum_{t\ge t_m} x_t c_t\ \le\ \rho\sum_{t\ge t_m} x_t v_t-bal .
\end{aligned}
\end{equation}
Any solution feasible for~\eqref{eq:orig} satisfies $\sum x_t c_t\le B_{t_m}\le B_{t_m}+bal$ and $\sum x_t c_t\le\rho\sum x_t v_t-bal\le\rho\sum x_t v_t$, hence is feasible for~\eqref{eq:reduced}. As the two problems share the same objective and~\eqref{eq:reduced}$\equiv$\eqref{jom} by Theorem~\ref{thm:equiv}, the feasible region of the original problem is contained in that of the joint problem, giving $V_O^\star\le V_J^\star$. The gap is governed by $bal$: the larger the historical violation, the more the joint formulation relaxes the future bidding problem.
\end{proof}

This establishes the central motivation of JD-BP: rather than forcing the bidding policy to repay historical KPI violations (the $bal$ penalty in~\eqref{eq:orig}), the pricing correction absorbs the deficit, decoupling future value maximization from historical constraint satisfaction.

\noindent\textbf{Market-level efficiency and revenue.}
The benefit above extends from a single advertiser to the whole auction market. Consider a CPM second-price market with advertisers indexed by $m\in\mathcal{M}$. By the uniform-bidding optimum of Eq.~\eqref{opt-bid}, which for the ROI constraint reduces to a value-proportional bid, each advertiser $m$ bids $b_{im}=\alpha_m v_{im}$ on impression $i$ with a single multiplier $\alpha_m>0$; we adopt the standard large-market (non-atomic) regime in which each advertiser treats the realized cost process as exogenous.

\begin{lemma}\label{lem:mult}
Repaying the deficit through bidding (single-action regime) yields a multiplier $\alpha_m^{O}(bal_m)\le\alpha_m^{\circ}$, where $\alpha_m^{\circ}$ is the deficit-free multiplier used by JD-BP; equality holds iff $bal_m=0$, and $\alpha_m^{O}$ is strictly decreasing in $bal_m$.
\end{lemma}
\begin{proof}
Relative to JD-BP, the single-action regime both lowers the budget (from $B_{t_m}+bal_m$ to $B_{t_m}$) and tightens the ROI bound by $bal_m$ (Eq.~\eqref{eq:orig}), so its feasible set is contained in JD-BP's (Corollary~\ref{cor:gap}). The bidding multiplier is the largest scale at which the binding constraint still holds; relaxing that constraint by removing $bal_m$ can only enlarge this scale, so $\alpha_m^{O}\le\alpha_m^{\circ}$, strictly when $bal_m>0$.
\end{proof}

Writing $\kappa_m=\alpha_m^{O}/\alpha_m^{\circ}\in(0,1]$, the single-action bid is $b_{im}^{O}=\kappa_m\,\alpha_m^{\circ}v_{im}$: the deficit-free bid $\alpha_m^{\circ}v_{im}$ perturbed by a \emph{value-independent}, advertiser-specific factor $\kappa_m$, which JD-BP removes ($\kappa_m\equiv1$).

\begin{theorem}[Allocation efficiency and revenue]\label{thm:market}
Let $W=\sum_i v_{i,w(i)}$ be the realized welfare ($w(i)$ is the winner of impression $i$) and $R$ the legitimate platform revenue, i.e.\ payments net of the returned historical over-collection $\sum_m bal_m$. Then JD-BP weakly improves both allocation efficiency and revenue: its allocation coincides with the deficit-free (constrained-efficient) one, and $R^{O}\le R^{J}$.
\end{theorem}
\begin{proof}
A \emph{common} positive scaling of all bids leaves every auction outcome $\arg\max_m b_{im}$ unchanged, so the single-action allocation departs from the deficit-free one only through the \emph{dispersion} of the value-independent factors $\{\kappa_m\}$: whenever these are non-degenerate, an impression can be diverted from the advertiser the deficit-free market would serve to one carrying a larger historical deficit---a misallocation uncorrelated with value. JD-BP sets $\kappa_m\equiv1$ and thus restores exactly the deficit-free, constrained-efficient allocation. For revenue, $\alpha_m^{\circ}\ge\alpha_m^{O}$ (Lemma~\ref{lem:mult}) weakly raises every bid, so each second price $c_i=\max_{k\ne w(i)}b_{ik}$---an order statistic, monotone in the bids---weakly increases; as the refund returns only the non-legitimate over-collection $\sum_m bal_m$, the legitimate component $\sum_i c_i$ is weakly larger, further reinforced by the fuller budget utilization of the memoryless design.
\end{proof}

Hence relocating constraint compensation to a non-distortionary pricing channel restores value-revealing bids, improving allocation efficiency and platform revenue \emph{simultaneously}.

\begin{figure*}[tp]
    \centering
    \includegraphics[width=1.0\linewidth]{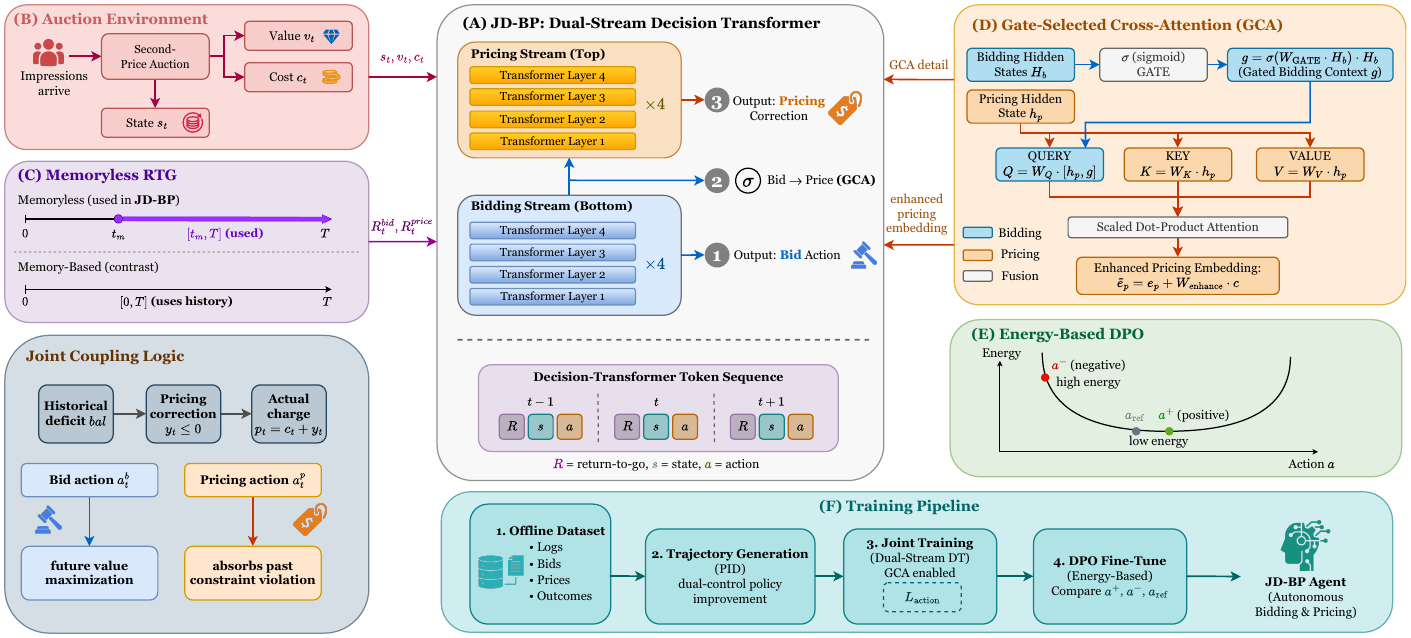}
      \caption{Overview of the JD-BP framework.
        \textbf{(a)} Dual-Stream Decision Transformer: parallel Bidding (blue) and Pricing (orange) streams over a return/state/action token sequence.
        \textbf{(b)} Auction Environment emitting per-step value $v_t$ and cost $c_t$.
        \textbf{(c)} Memoryless RTG conditioned only on the future $[t_m, T]$.
        \textbf{(d)} GCA: a gate $\sigma$ fuses bidding context into the pricing query.
        \textbf{(e)} Energy-Based DPO placing $a^+$/$a^-$ at low/high energy w.r.t.\ $a_{\mathrm{ref}}$.
        \textbf{(f)} Offline training pipeline ending at the deployed JD-BP agent.}
    \label{fig:algorithm_stru}
\end{figure*}

\subsection{Trajectory Generation}
In conventional DT based auto-bidding systems, trajectories follow the standard formulation:
\[
\tau = \{(r_t, s_t, a_t^{\text{bid}})\}_{t=1}^T
\]
where \(s_t \in \mathcal{S}\) represents the state of the system, $a_t^{\text{bid}} \in \mathcal{A}_{\text{bid}}$ denotes the bidding action, and $r_t \in \mathbb{R}$ is the RTG.

In our joint decision-making framework, we extend this formulation to incorporate pricing correction actions, resulting in an augmented trajectory structure:
\[
\tau_{\text{joint}} = \{(R_t^{\text{bid}}, R_t^{\text{price}}, s_t, a_t^{\text{bid}}, a_t^{\text{price}}, A_t)\}_{t=1}^T
\]
where \(R_t^{\text{bid}}\) and \(R_t^{\text{price}}\) represent the target returns for bidding and pricing objectives respectively, and \(a_t^{\text{price}} \in \mathcal{A}_{\text{price}}\) denotes the pricing correction action.

While online interaction with the environment provides one mechanism for data collection, we propose a more efficient offline trajectory generation procedure based on the optimal solution derived from our theoretical analysis. Specifically, we randomly select a trajectory from the set generated by the base bidding policy. Assuming there exists a constraint violation $bal_t$ at time step $t$ and we aim to compensate for it within the remaining steps, we compute $a_t^{\text{price}}$ via a PID controller~\cite{Ziegler1942OptimumSF}, chosen for its simplicity, interpretability, and proven efficacy at dynamically adjusting actions to improve target achievement while producing diverse, high-quality pricing-action data; any algorithm meeting these criteria could be substituted. This allows us to update the subsequent state and, based on this updated state and the base bidding policy, generate the next bidding action. We repeat the above steps until a complete trajectory is generated and calculate the corresponding $R_t^{\text{bid}}$ and $R_t^{\text{price}}$. \par
Simultaneously, at each step of the loop, we additionally generate a new trajectory without pricing actions following the base bidding policy. This enables computation of the advantage between applying versus not applying pricing actions in the current state, thereby enhancing the model's exploration capability during training. We detail the trajectory generation process in Algorithm~\ref{alg-trg}.

\subsubsection{Memoryless RTG}
In our model, the introduction of pricing actions decouples bidding decisions from historical constraint violations, allowing the bidding policy to focus solely on future value maximization under future constraint satisfaction. Consequently, differing from the conventional definition of RTG in prior DT-related research, we introduce a memoryless RTG to guide the model in generating bidding actions that prioritize future constraint fulfillment and value maximization. The memoryless RTG formulations for bidding and pricing are defined as:
\begin{equation}
R_t^b =\min\Big(\big(\frac{\sum_{i=t}^T v_i}{\sum_{i=t}^T{c_i}}\big)^2, 1\Big) \cdot \sum_{i=t}^T v_i \\
\end{equation}

\begin{equation}
R_{t}^p = \phi_t^{\,\omega_t}
\end{equation}

\begin{equation}
\phi_t = \min\!\Big(\frac{\sum_{i=1}^T v_i}{\sum_{i=1}^T \hat{c}_i},\ \frac{\sum_{i=1}^T \hat{c}_i}{\sum_{i=1}^T v_i}\Big)
\end{equation}

\begin{equation}
\omega_t = 2\left|\frac{\sum_{i=t}^T y_i}{\sum_{i=1}^T v_i-\sum_{i=1}^T c_i-\sum_{i=1}^t y_i}-1\right|+1
\end{equation}

where $\hat{c}_i = c_i + y_i$ is the post-correction charge, and $c_i$, $v_i$, $y_i$ denote the charge before pricing correction, the value of the $i$-th impression, and the pricing correction, respectively.
The $R_t^b$ incorporates the constraint-fulfillment penalty over the total advertising value after time step $t$. $R_t^p$ is designed so that the post-correction cost closely approximates the advertising value while satisfying the cost constraint: the base $\phi_t$ measures the alignment between post-correction cost and value over the whole period, while the exponent $\omega_t$ assesses the ability to recover both current and future pricing deviations by the period's end. This design ensures $R_t^p\in[0,1]$. \par
Note that $R_t^b$ uses the pre-correction cost rather than the post-correction cost: relying on the latter would let bidding over-spend while leaning on pricing to satisfy constraints, degrading allocation efficiency. Masking the post-correction status thus keeps bidding economically reasonable and preserves allocation efficiency.

\begin{algorithm} 
    \caption{Trajectory Generation Algorithm}
    \label{alg-trg}
    \begin{algorithmic}[1] 
    \REQUIRE Original Offline Dataset $\mathcal{D}_0$, Base policy with bidding action only $\pi_0$, correction window $h$
    \ENSURE $\mathcal{D}_{\text{joint}}$
    
    \STATE Select $\tau = \{(s_t, R_t^b, a_t^b)\}_{t=0}^T \in \mathcal{D}_0$
    \STATE Generate intermediate trajectory $\tau_{\text{int}} = \{(R_t^b, \hat{s}_t, \hat{a}_t^b, a_t^p, A_t)\}$
    
    \FOR{$t = 0, 1, \ldots, T$}
        \STATE Compute $a_t^p$
        \STATE State transition: $\hat{s}_{t+1} = f_{\text{env}}(\hat{s}_t, \hat{a}_t^b, a_t^p)$
        \STATE Generate the bidding action by the base policy: $\hat{a}_{t+1}^b = \pi_0\left( [\tau[0: t], \hat{s}_{t+1}] \right)$
        \FOR{$l = t, 1, \ldots, T$}
            \STATE State transition without pricing action: $s_{t+1} = f_{\text{env}}(\hat{s}_t, \hat{a}_t^b, \textbf{0})$
            \STATE Generate the bidding action by the base policy: $a_{t+1}^b = \pi_0\left( [\tau[0: t], s_{t+1}] \right)$
        \ENDFOR
        \STATE Compute $R_t^b(\hat{s}_t)$ under state $\hat{s}_t$ without employing pricing actions henceforth.
    \ENDFOR
    
    \FOR{$t = 0, 1, \ldots, T$}
        \STATE Compute $R_t^b, R_t^p$, and advantage value $A_t = R_t^p - R_t^b(\hat{s}_t)$
    \ENDFOR

    \STATE Generate augmented trajectory: $$\tau_{\text{joint}} = \{(R_t^b, R_t^p, \hat{s}_t, \hat{a}_t^b, a_t^p, A_t)\}$$
    
    \end{algorithmic}
\end{algorithm}

\subsection{Joint Decision Transformer with Energy-Based DPO Fine-tuning}
DT has been used as the backbone model in several generative auto-bidding works for its capability of capturing trajectory-wise information.
However, the self-attention module in causal transformer grants the bidding agent access to past KPI violation conditions, leading to allocation inefficiency.
Moreover, extended trajectories generated in the offline environment may cause further distribution shifting. 
To address these two challenges, we propose a joint decision transformer model with energy-based DPO fine-tuning, where the joint model is first trained with a supervised action loss and then fine-tuned with the energy-based DPO objective on the preference pairs produced by Algorithm~\ref{alg-trg}.

\subsubsection{Joint Decision Transformer}
Let the extended sequence $\tau \in \mathcal{D}_{joint}$ at time \(t\) be represented as:
\[
\tau_t = (s_t, R_t^{\text{b}}, R_t^{\text{p}}, A_t, a_t^{\text{b}}, a_t^{\text{p}}).
\] 
We first divide the extended trajectory \(\tau_t\) at time \(t\) into a \textbf{bidding trajectory} \(\tau^b_t\) and a \textbf{pricing trajectory}
\(\tau^p_t\) where 
\begin{equation}
\tau^b_t = (s_t, R_t^b, a_t^b),\,\,\, \tau^p_t = (s_t, R_t^p, a_t^p).
\end{equation}

The bidding trajectory is fed into a standard causal transformer~\cite{Vaswani2017AttentionIA}: $\mathbf{H}_b^{1:t}=\text{Transformer}_b([\mathbf{e}_1^b,\dots,\mathbf{e}_t^b])$, and the bid is read out from the last hidden state $\mathbf{h}_t^b=\mathbf{H}_b^{1:t}[-1]$ as $a_t^b=W_{\text{out}}^b\,\text{LayerNorm}(\mathbf{h}_t^b)$.

Similarly, the pricing trajectory is processed by a causal transformer followed by a cross attention module with the last hidden state of bidding transformer.

\noindent\textbf{Gate-Selected Cross-Attention for Coordinated Correction}
The pricing stream dynamically adapts based on bidding outcomes through GCA:
\begin{align}
\mathbf{g}_{t}^b &= \sigma(W_{GATE}(\mathbf{H}_b^{1:t})) \cdot \mathbf{H}_b^{1:t}\\
\mathbf{Q}_p &= W_Q^{\text{cross}} \cdot [\mathbf{h}_{t}^p, {g}_{t}^b] \\
\mathbf{K}_p &= W_K^{\text{cross}} \cdot \mathbf{h}_{t}^p \\
\mathbf{V}_p &= W_V^{\text{cross}} \cdot \mathbf{h}_{t}^p \\
\mathbf{c}_t &= \text{Softmax}\left(\frac{\mathbf{Q}_p\mathbf{K}_p^\top}{\sqrt{d_k}}\right)\mathbf{V}_p \\
\tilde{\mathbf{e}}_t^p &= \mathbf{e}_t^p + W_{\text{enhance}} \cdot \mathbf{c}_t
\end{align}
Here $\sigma$ is the sigmoid gate, letting the pricing module attend to the full bidding context rather than only the final bid.

The pricing stream mirrors the bidding stream on the enhanced embedding: $\mathbf{H}_p^{1:t}=\text{Transformer}_p([\mathbf{e}_1^p,\dots,\mathbf{e}_{t-1}^p,\tilde{\mathbf{e}}_t^p])$ and $a_t^p=W_{\text{out}}^p\,\text{LayerNorm}(\mathbf{h}_t^p)$.
In this dual-stream design, bidding decisions are made first and pricing corrections are conditioned on the realized bidding outcomes; this asymmetric coupling lets the pricing module adapt to the acquired traffic (premium strategies for high-value traffic, corrective measures otherwise) without contaminating the core bidding logic that sustains impression volume and market competitiveness.

To address the joint learning of bidding and pricing actions, we employ a supervised regression loss to align the model's outputs with the ground truth actions. The loss function is defined as:
\begin{equation}
\mathcal{L}_{\text{action}} = \frac{1}{T} \sum_{t=1}^T \left[ \eta_b  \cdot \| a_t^b - a_t^{b^*} \|^2 + \eta_p  \cdot \| a_t^p - a_t^{p^*} \|^2 \right],
\end{equation}
where $a_t^b$ and $a_t^p$ denote the model's predicted bidding and pricing actions at time step $t$, $a_t^{b^*}$ and $a_t^{p^*}$ are the corresponding ground-truth actions, and $\eta_b$, $\eta_p$ are weighting coefficients.

\subsubsection{Energy-Based DPO Fine-tuning.}
\par
We collect sample pairs for DPO fine-tuning using the procedure described in Algorithm~\ref{alg-trg}. At each decision step, we determine the positive (winner) sample $a^+$ and negative (loser) sample $a^-$ based on the advantage value $A_t = R_t^p - R_t^b(\hat{s}_t)$. Specifically, if the RTG increases after introducing the pricing action, we label the corresponding action as positive; otherwise, it is labeled as negative.

While standard DPO was developed for stochastic policies (LLMs) via categorical distributions, our bidding and pricing framework utilizes a deterministic policy that directly outputs continuous scalar values. To apply DPO in this deterministic regression setting without imposing artificial distribution assumptions on the model architecture, we derive the objective through an \textit{energy-based formulation}~\cite{LeCun2006ATO,Du2019ImplicitGA}.

\paragraph{Energy-Based Theoretical Derivation.}
In the context of continuous control, we define an energy function $E(a, y_{target})$ that quantifies the "cost" or incompatibility between a predicted action $a$ and a target value $y_{target}$. For our task, we adopt the L1 distance as the energy metric: $E(a, y_{target}) = |a - y_{target}|$. Lower energy implies higher compatibility. Under the Boltzmann distribution assumption commonly used in energy-based models, the implicit preference probability is proportional to the negative energy, i.e., $P(y_{target}|a) \propto \exp(-\frac{1}{\beta}E(a, y_{target}))$. Consequently, the standard DPO term, which represents the log-ratio of the policy to the reference model, can be reformulated as the relative energy gain:
\begin{equation}
\begin{aligned}
\log \frac{\pi_\theta(a|\cdot)}{\pi_{\text{ref}}(a|\cdot)} &\equiv -\frac{1}{\beta} \left( E(a, y_{target}) - E(a_{\text{ref}}, y_{target}) \right) \\
&= \frac{1}{\beta} \left( |a_{\text{ref}} - y_{target}| - |a - y_{target}| \right)
\end{aligned}
\end{equation}

This recasts the probabilistic DPO objective as a deterministic distance-minimization problem: the model $\pi_\theta$ is pushed to reduce its energy (error) relative to the frozen reference $\pi_{\text{ref}}$.

\paragraph{Loss Function.}
Writing the relative energy gain as a similarity score $S(a, y_{target}) = | a_{\text{ref}} - y_{target} | - | a - y_{target} |$ and substituting it into the Bradley--Terry model used by DPO, the loss for a positive/negative pair $(a^+, a^-)$ becomes:
\begin{equation}
    \mathcal{L}_{\text{DPO}} = -\mathbb{E}_{(s, a^+, a^-) \sim \mathcal{D}} \left[ \log \sigma \big( \beta \left( S(a, a^+) - S(a, a^-) \right) \big) \right],
\end{equation}
where $a = \mu_\theta(s)$ is the model's continuous scalar output; the objective drives $a$ toward $a^+$ and away from $a^-$, aligning the policy with high-reward regions.

\subsubsection{A Model-Agnostic Plug-in Framework}
We emphasize that JD-BP is not tied to the Decision Transformer backbone used in this paper; it is a \emph{plug-in, model-agnostic} framework that can be instantiated on top of stronger generative auto-bidding architectures such as GAS, GAVE, or GRAD. As established in Theorem~\ref{thm:equiv} and Corollary~\ref{cor:gap}, the benefit of decoupling value-maximizing bidding from historical constraint satisfaction is a property of the \emph{problem formulation} rather than of any specific policy parameterization---the equivalence and the $V_O^\star\le V_J^\star$ gap hold regardless of how the bidding policy is realized. Concretely, porting JD-BP to another backbone requires only three architecture-independent ingredients: (1) augmenting the action space with the pricing-correction term and designing the \emph{bid$\rightarrow$price interaction} (here realized by GCA, but any asymmetric coupling that conditions pricing on bidding outcomes applies); (2) replacing the conditioning/return signal with a \emph{memoryless} design so that future decisions are not dragged down by past KPI violations; and (3) optionally appending the energy-based DPO refinement. Because none of these components depends on the underlying sequence model, JD-BP can be layered onto existing state-of-the-art bidders to confer the same joint bidding--pricing capability.

\begin{table*}[t]
\centering
\renewcommand{\arraystretch}{1.0}
\setlength{\tabcolsep}{7pt}
\begin{tabular}{llcccccccccc|c}
\toprule
\textbf{Dataset} & \textbf{Budget} & \textbf{DiffBid} & \textbf{USCB} & \textbf{CQL} & \textbf{IQL} & \textbf{BCQ} & \textbf{DT} & \textbf{CDT} & \textbf{GAS} & \textbf{GAVE} & \textbf{GRAD} & \textbf{JD-BP} \\
\midrule
\multirow{5}{*}{AuctionNet}
 & 50\%  & 9.9  & 11.5 & 12.8 & 16.5 & 17.7 & 14.8 & 11.2 & 18.4 & 19.6 & \underline{20.0} & \textbf{21.2} \\
 & 75\%  & 15.4 & 14.9 & 16.7 & 22.1 & 24.6 & 22.9 & 18.0 & 27.5 & 28.3 & \underline{28.5} & \textbf{32.9} \\
 & 100\% & 19.5 & 17.5 & 22.2 & 30.0 & 31.1 & 29.6 & 31.2 & 36.1 & 37.2 & \underline{37.4} & \textbf{42.5} \\
 & 125\% & 25.3 & 26.7 & 28.6 & 37.1 & 34.2 & 34.3 & 31.7 & 40.0 & 42.7 & \underline{43.2} & \textbf{44.9} \\
 & 150\% & 30.8 & 31.3 & 35.8 & 43.1 & 37.9 & 44.5 & 39.1 & 46.5 & 47.4 & \underline{47.5} & \textbf{47.5} \\
\midrule
 & \textbf{Avg} & 20.2 & 20.4 & 23.2 & 29.8 & 29.1 & 29.2 & 26.2 & 33.7 & 35.0 & \underline{35.3} & \textbf{37.8} \\
\bottomrule
\end{tabular}
\caption{
Score comparison on AuctionNet (period-7, 48 advertisers) under five budget levels (50\%--150\% of the original budget). The highest score in each row is in bold; the best baseline is underlined. Higher is better.
}
\label{tab:scoreresults}
\end{table*}

\section{Offline Experiments}
\subsection{Setup}
\subsubsection{Datasets}
We evaluate our approach on the Alibaba open-source AuctionNet dataset~\cite{Su2024AuctionNetAN}, the largest publicly available resource in the auto-bidding domain, which records the prior estimated value and bid price of each advertiser in every ad request. Multiple requests can be aggregated over time intervals to reconstruct complete bidding trajectories. Since the original data does not contain pricing actions, we employ a PID controller to dynamically adjust the pricing process so that the true CPA fits the target tCPA. As established in our problem formulation in Eq.~\eqref{ori-pro}, optimizing for tCPA is mathematically equivalent to our formulated tROI constraint by taking the reciprocal of the target value, which allows us to directly apply our joint-decision framework to the tCPA-based dataset and to collect the bidding and pricing trajectories required for training.

Each ad request involves 48 advertisers competing under a CPM-based second-price mechanism. Following the evaluation protocol adopted by recent generative auto-bidding works~\cite{Gao2025GenerativeAW}, we report results on the \textbf{AuctionNet} setting, using the period-7 traffic with all 48 advertisers as the test set. During testing, we iterate through the 48 advertisers, assigning the algorithm under evaluation to each one in turn while the remaining 47 use the baseline strategy. To comprehensively assess performance across operating regimes, we evaluate every method under \textbf{five budget ratios} $\{0.5, 0.75, 1.0, 1.25, 1.5\}$ applied to the original advertiser budgets, and report the score at each ratio together with the average.

\subsubsection{Parameter Settings}
During training, we use the PyTorch framework on a single NVIDIA H100 GPU. The model is optimized with AdamW, a learning rate of $1\mathrm{e}{-4}$, and a weight decay of $1\mathrm{e}{-4}$. Each of the two decision streams (bidding and pricing) is a 4-layer Transformer with a hidden size of 128, 4 attention heads, and a feed-forward dimension of 512; the context length is set to $K=20$, giving roughly 2.1M parameters in total. The state dimension is 23, including features related to cost before correction, and the value weighting coefficients $\eta_b$ and $\eta_p$ are both set to 1. For DPO fine-tuning, we set $\beta$ to 0.15 and train for 3 epochs with a learning rate of $1\mathrm{e}{-5}$. We filter out trajectories whose advantage $A_t$ equals zero, and further exclude samples at decision step $t$ whose future target cost (computed as $\mathrm{TCPA} \cdot \sum_{i=t}^T v_i$) is less than $0.5$.

\subsubsection{Metrics}
We adopt the evaluation metric officially defined by AuctionNet, where the score combines value maximization with a penalty for violating the cost constraint:
\begin{equation}
Score =\min\Big(\big(\frac{\sum_{i=1}^T v_i}{\sum_{i=1}^T{c_i}}\big)^2, 1\Big) \cdot \sum_{i=1}^T v_i.
\end{equation}
A higher score is better.

\subsubsection{Baselines}
Following the protocol of recent works, we compare against representative offline-RL and generative auto-bidding methods, using the scores reported under the same AuctionNet setting.
\textbf{RL-based:} \textbf{USCB}~\cite{he2021unified} unifies bidding-parameter optimization via online RL; \textbf{CQL}~\cite{Kumar2020ConservativeQF} learns a conservative value function to curb overestimation; \textbf{IQL}~\cite{Kostrikov2021OfflineRL} uses expectile regression to avoid evaluating out-of-distribution actions; \textbf{BCQ}~\cite{fujimoto2019off} constrains the policy to dataset-supported actions.
\textbf{Generative:} \textbf{DiffBid}~\cite{Guo2024AIGBGA} generates bidding trajectories with diffusion models; \textbf{DT}~\cite{Chen2021DecisionTR} casts decision-making as conditional sequence modeling; \textbf{CDT}~\cite{liu2023constrained} extends DT with explicit constraint conditioning; \textbf{GAS}~\cite{li2025gas} augments sequence modeling with post-training search; \textbf{GAVE}~\cite{Gao2025GenerativeAW} adds value-guided exploration; \textbf{GRAD}~\cite{lei2026generative} is the state-of-the-art generative method and our strongest baseline.

\subsection{Overall Performance}
Table~\ref{tab:scoreresults} reports the score of JD-BP and all baselines across the five budget ratios. JD-BP attains the highest average score of \textbf{37.8}, outperforming the strongest baseline GRAD (35.3) by \textbf{+2.5}, and achieves the best (or tied-best) result in every budget setting. The advantage is most pronounced in the mid-budget regimes (e.g., 42.5 vs.\ 37.4 at ratio 1.0 and 44.9 vs.\ 43.2 at ratio 1.25), where the joint bidding--pricing policy can both capture high-value traffic and recover cost through pricing correction. Notably, the gain holds consistently from tight (0.5) to loose (1.5) budgets, demonstrating that JD-BP generalizes across operating regimes rather than excelling only in a single setting. These results confirm that explicitly modeling the pricing action as a controllable decision, on top of bidding, yields a substantial and robust performance improvement over both reinforcement-learning and generative baselines.

\begin{table}[t]
\centering
\renewcommand{\arraystretch}{1.1}
\resizebox{\columnwidth}{!}{%
\begin{tabular}{lcccccc}
\toprule
\textbf{Variant} & \textbf{50\%} & \textbf{75\%} & \textbf{100\%} & \textbf{125\%} & \textbf{150\%} & \textbf{Avg} \\
\midrule
\textbf{JD-BP (Full)} & \textbf{21.2} & \textbf{32.9} & \textbf{42.5} & \textbf{44.9} & \textbf{47.5} & \textbf{37.8} \\
\quad w/o DPO        & 18.6 & 27.6 & 38.9 & 43.7 & 47.4 & 35.2 \\
\quad w/o GCA        & 17.8 & 31.3 & 37.7 & 40.6 & 44.2 & 34.3 \\
\quad w/o m-RTG      & 20.7 & 29.4 & 37.4 & 38.3 & 43.5 & 33.8 \\
\bottomrule
\end{tabular}%
}
\caption{
Ablation of JD-BP on AuctionNet. Each row removes one component from the full model: DPO, the GCA module, or the memoryless RTG (m-RTG). Higher is better.
}
\label{tab:ablation}
\end{table}

\subsection{Ablation Studies}
To dissect the contribution of each component, we ablate three key designs of JD-BP---the memoryless RTG, the Gate-Selected Cross-Attention (GCA) module, and the DPO enhancement---and report the results in Table~\ref{tab:ablation}. Removing any single component degrades the average score monotonically, confirming that all three are effective; their importance ranks as memoryless RTG $>$ GCA $>$ DPO.

\noindent\textbf{Effect of the memoryless RTG (w.\ hisRTG):}
    This variant replaces the memoryless return-to-go with a history-aware RTG that subtracts the value already achieved, yielding a more conservative bidding signal:
    \begin{equation}
    R_t^b =\min\Big(\big(\tfrac{\sum_{i=1}^T v_i}{\sum_{i=1}^T{c_i}}\big)^2, 1\Big) \cdot \sum_{i=1}^T v_i - \min\Big(\big(\tfrac{\sum_{i=1}^{t-1} v_i}{\sum_{i=1}^{t-1}{c_i}}\big)^2, 1\Big) \cdot \sum_{i=1}^{t-1} v_i.
    \end{equation}
    It causes the largest degradation, dropping the average from 37.8 to 33.8 ($-4.0$). The history-aware signal makes the policy increasingly conservative as the episode proceeds, leading to budget under-utilization in the mid-to-high budget regimes (e.g., only 38.3 at ratio 1.25 vs.\ 44.9 for the full model). This validates the memoryless design, which keeps the target signal stable and lets the model fully exploit the available budget.

\noindent\textbf{Effect of the GCA module (w.o.\ GCA):}
    Removing the Gate-Selected Cross-Attention drops the average from 37.8 to 34.3 ($-3.5$). GCA explicitly conditions the pricing stream on the bidding outcomes, establishing a causal, asymmetric coupling (Bid $\rightarrow$ Price). Without it, the two decision streams are fully decoupled and cannot coordinate, so the pricing correction is less aligned with the realized bids. The consistent drop across all budgets shows that this explicit coupling, rather than relying on entangled implicit interactions, is important for robust joint decision-making.

\noindent\textbf{Effect of the DPO enhancement:}
    Comparing the base model with the full model isolates the contribution of preference optimization. DPO improves the average score from 35.2 to 37.8 ($+2.6$, $+7.4\%$), with the gain concentrated in the tighter budget regimes (e.g., 32.9 vs.\ 27.6 at ratio 0.75). This indicates that DPO effectively aligns the pricing policy toward high-efficiency trajectories, providing gains that are orthogonal to the architectural designs.

In summary, the ablation confirms our key design choices: the memoryless RTG enables full budget utilization, the GCA module couples pricing with bidding for coordinated decisions, and the DPO stage delivers further gains, together forming the complete JD-BP model.

\begin{figure}[t]
    \centering
    \includegraphics[width=0.9\linewidth]{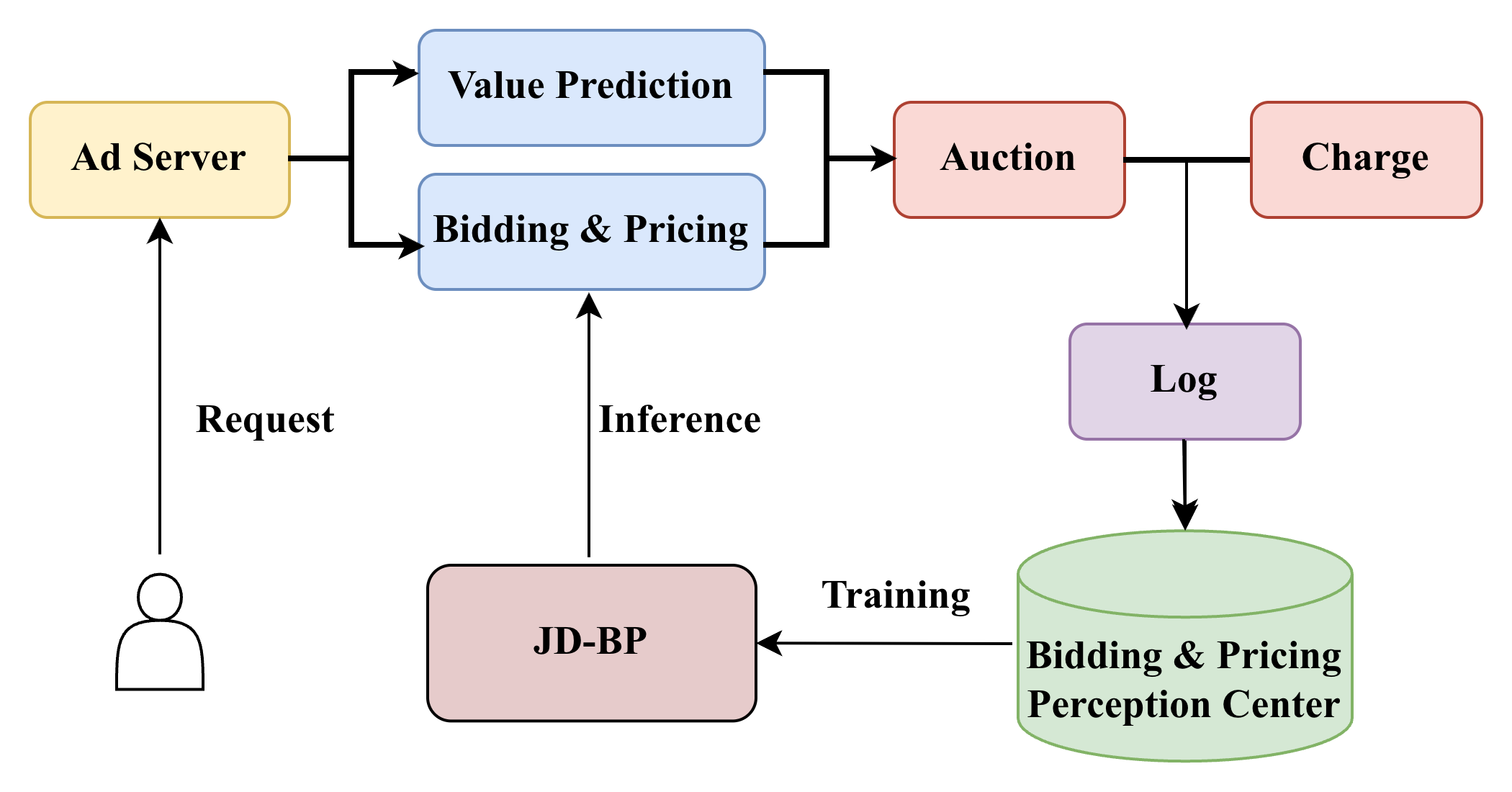}
    \caption{Online deployment workflow of the bidding and pricing system. User requests are processed by the Ad Server, followed by value prediction and bidding \& pricing modules. The results are submitted to the auction and charging components, and all operational data are recorded and monitored by the Bidding \& Pricing Perception Center for continuous optimization.}
    \label{fig:online_stru}
\end{figure}

\subsection{Online Deployment}
To validate JD-BP in real-world industrial systems, we deployed it on an advertising platform and conducted large-scale online A/B experiments; the online architecture is shown in Fig.~\ref{fig:online_stru}. On this platform, the revenue from target-bidding ads reaches tens of millions of RMB, which is sufficient to ensure reliable results. Because JD-BP introduces a pricing action for which no data exists in the initial state, the system was rolled out in \emph{two stages}, which conveniently doubles as an online ablation of the generative pricing model. \textbf{Stage~1} pairs a memoryless bidding policy with a \emph{rule-based PID pricing} controller, which also accumulates the joint bidding--pricing trajectories used to train JD-BP; \textbf{Stage~2} keeps the same memoryless bidding and replaces \emph{only} the PID controller with the full \emph{generative-joint} pricing of JD-BP, so that the Stage~1$\rightarrow$Stage~2 difference is a controlled ablation isolating the learned pricing. As Table~\ref{tab:online_results} shows, Stage~1 already improves substantially over the pre-deployment base, and, crucially, Stage~2 yields a further clear gain \emph{on top of this strong PID-pricing system} ($+1.38\%$ target cost, $+1.30\%$ ad revenue, $+0.29$pp achievement). This directly shows that the learned generative pricing contributes a distinct, positive increment beyond rule-based PID pricing in production---evidence for the necessity of the generative model that a rule-based controller alone cannot provide. Notably, Stage~1 already embodies two of our own contributions---the pricing-correction action and memoryless bidding---so its large gain itself validates the paper's core decoupling idea rather than any external baseline, while Stage~2 isolates the further value of \emph{learning} the joint policy generatively over a hand-tuned PID controller. Cumulatively, JD-BP improves ad revenue by $4.70\%$, target cost by $6.48\%$, and achievement rate by $3.92$pp over the base system.

\begin{table}[htbp]
    \centering
    \caption{Staged online A/B deployment, also serving as an online ablation of the generative pricing model. Stage~1 (memoryless bidding + rule-based PID pricing) is measured against the pre-deployment base; Stage~2 (the full generative-joint JD-BP, replacing PID pricing) is measured against Stage~1 and isolates the generative model's contribution; the last row is JD-BP's cumulative effect over the base. Achievement is the fraction of ads with CPA/TCPA in $[0.8,1.2]$.}
    \label{tab:online_results}
    \small
    \resizebox{\columnwidth}{!}{%
    \begin{tabular}{lccc}
        \toprule
        \textbf{Deployment stage} & \textbf{Ad Revenue} & \textbf{Target Cost} & \textbf{Achievement} \\
        \midrule
        Stage 1: Memoryless bidding + PID pricing & $+3.40\%$ & $+5.10\%$ & $+3.63$pp \\
        Stage 2: + Generative joint (JD-BP) & $+1.30\%$ & $+1.38\%$ & $+0.29$pp \\
        \midrule
        Overall (JD-BP vs.\ base) & $+4.70\%$ & $+6.48\%$ & $+3.92$pp \\
        \bottomrule
    \end{tabular}%
    }
\end{table}

\section{Related Works}
\subsection{Offline Reinforcement Learning for Auto-Bidding}
Reinforcement learning (RL) optimizes sequential decisions through interaction with an environment~\cite{Sutton1998ReinforcementLA}, with classic algorithms such as DQN~\cite{Mnih2015HumanlevelCT} and PPO~\cite{Schulman2017ProximalPO} succeeding via online exploration. In auto-bidding, however, direct exploration is costly or infeasible, so offline RL~\cite{Levine2020OfflineRL, Agarwal2019AnOP} learns policies solely from logged data such as historical impressions and conversion logs. To counter its characteristic distributional shift and overestimation bias, methods such as BCQ~\cite{Fujimoto2018OffPolicyDR}, CQL~\cite{Kumar2020ConservativeQF}, and MOPO~\cite{Yu2020MOPOMO} constrain policy updates, add conservative objectives, or learn environment models, enabling the safe and efficient policy improvement that auto-bidding requires. Building on these advances, a growing body of work applies RL directly to auto-bidding, casting budget pacing and constraint satisfaction as sequential control~\cite{Cai2017RealTimeBB,Ye2020DeepRL,mou2022sustainable} and unifying multiple KPI constraints within a single learned bidding controller~\cite{he2021unified}. These methods established RL as a competitive paradigm for constrained bidding, yet they still optimize the bid as the sole action under a fixed charging rule. Beyond these general-purpose algorithms, several works tailor offline RL to advertising, learning value critics to guide bid scaling under budget constraints or imposing conservative and expectile objectives to stay within the logged distribution~\cite{Kostrikov2021OfflineRL,Kumar2020ConservativeQF}; yet such value-based methods remain brittle under the long horizons and the sparse, delayed conversions of daily campaigns, which partly motivates the sequence-modeling formulations discussed next.

\subsection{Generative Methods}
Recently, generative models have demonstrated significant potential in the field of automated bidding. Mainstream approaches include Variational Autoencoders (VAE)~\cite{Kingma2013AutoEncodingVB}, diffusion models~\cite{Ho2020DenoisingDP}, and sequence modeling architectures such as Decision Transformer~\cite{Chen2021DecisionTR} and Trajectory Transformer~\cite{Janner2021OfflineRL}, which can effectively represent complex distributions or conditional relationships in bidding environments. Transformer-based frameworks leverage autoregressive mechanisms to capture high-dimensional dependencies within advertising platforms, as exemplified by models like GAVE~\cite{Gao2025GenerativeAW} and GAS~\cite{Li2024GASGA}. In parallel, diffusion models generate high-quality bidding samples through iterative conditional denoising processes~\cite{Guo2024AIGBGA,Li2025GenerativeAI,Peng2025ExpertGuidedDP}. These generative methods offer new avenues for optimizing bidding strategies and provide robust solutions to practical challenges like data sparsity and dynamic market conditions. Among them, Decision-Transformer variants are especially attractive for bidding: by conditioning on a target return and decoding actions autoregressively, they capture a campaign's temporal structure without the bootstrapped value estimation that often destabilizes value-based offline RL~\cite{Chen2021DecisionTR,liu2023constrained}.

More recent efforts further scale and refine this paradigm: GRAD~\cite{lei2026generative} trains a large-scale pre-trained generative model for ad bidding, GAS~\cite{Li2024GASGA} introduces a post-training search stage to refine generated strategies, CDT~\cite{liu2023constrained} augments the Decision Transformer with explicit constraint conditioning, and online-RL approaches such as USCB~\cite{he2021unified} unify the optimization of bidding parameters. Despite their architectural differences, all of these methods cast auto-bidding as learning a \emph{single} bidding action that must simultaneously maximize value and satisfy KPI constraints. In contrast, JD-BP is, to the best of our knowledge, the first to treat the charging mechanism as an additional controllable action, decoupling value maximization from historical constraint compensation. This design is orthogonal to the choice of backbone and can be layered on top of the above models.

\section{Conclusion}
We present JD-BP, a joint generative framework that optimizes both the bid and an additive pricing correction for online advertising auctions, addressing the misalignment caused by model uncertainty and feedback latency. By absorbing historical constraint violations through a memoryless Return-to-Go and the pricing channel, and refining the policy with energy-based DPO, JD-BP keeps bidding value-revealing while still satisfying KPI constraints. We further show, both theoretically and empirically, that this decoupling is never detrimental: relocating constraint compensation to a non-distortionary pricing channel provably improves the joint objective and, at the market level, enhances allocation efficiency, social welfare, and platform revenue simultaneously rather than trading one against another. Because these gains arise from the problem formulation rather than any specific backbone, JD-BP can be layered as a model-agnostic plug-in on top of stronger generative bidders. Experiments on the offline AuctionNet benchmark and online A/B tests on a large e-commerce platform confirm consistent gains in ad revenue and target cost, demonstrating its practical value for real-world bidding and pricing. Several directions remain open: extending the scalar refund into a richer pricing-correction action space, coupling the framework more tightly with alternative auction mechanisms such as reserve-price or first-price settings, and instantiating JD-BP on stronger generative backbones to further raise the ceiling demonstrated here. More broadly, we believe that treating the charging mechanism as a learnable decision opens a widely applicable axis for auto-bidding that extends well beyond the specific model studied in this work.

\printbibliography
\end{document}